\documentclass[runningheads,a4paper,orivec,envcountsect,envcountsame]{llncs}
\usepackage[utf8]{inputenc}
\usepackage[T1]{fontenc}
\usepackage[british]{babel}
\usepackage{versions}

\usepackage[%
rm={oldstyle=false,proportional=true},%
sf={oldstyle=false,proportional=true},%
tt={oldstyle=false,proportional=true,variable=true},%
qt=false%
]{cfr-lm}

\usepackage{amsmath,amssymb,mathrsfs,mathtools,wasysym,cite,graphicx,csquotes,subcaption,xfrac}
\usepackage[all,cmtip]{xy}
\usepackage[inline]{enumitem}
\setlist[enumerate,1]{label={\emph{\roman*)}}}

\usepackage{booktabs}
\usepackage[table]{xcolor}
\usepackage[textsize=small]{todonotes}
\numberwithin{equation}{section}
\usepackage[capitalise]{cleveref}
\usepackage{scalerel}
\usepackage{microtype}

\let\doendproof\endproof
\renewcommand\endproof{~\hfill\qed\doendproof}

\spnewtheorem*{remark*}{Remark}{\bfseries}{\rmfamily}
\spnewtheorem{assumption}{Assumption}{\bfseries}{\rmfamily}
\spnewtheorem{terminology}{Terminology}{\bfseries}{\rmfamily}

\crefname{section}{Sect.}{Sect.}
\Crefname{section}{Section}{Sections}
\crefname{figure}{Fig.}{Fig.}
\Crefname{figure}{Figure}{Figures}

\DeclareFontFamily{U}{matha}{\hyphenchar\font45}
\DeclareFontShape{U}{matha}{m}{n}{ <-6> matha5 <6-7> matha6 <7-8>
matha7 <8-9> matha8 <9-10> matha9 <10-12> matha10 <12-> matha12 }{}
\DeclareSymbolFont{matha}{U}{matha}{m}{n}
\DeclareFontFamily{U}{mathx}{\hyphenchar\font45}
\DeclareFontShape{U}{mathx}{m}{n}{ <-6> mathx5 <6-7> mathx6 <7-8>
mathx7 <8-9> mathx8 <9-10> mathx9 <10-12> mathx10 <12-> mathx12 }{}
\DeclareSymbolFont{mathx}{U}{mathx}{m}{n}

\DeclareMathDelimiter{\llbracket}{4}{matha}{"76}{mathx}{"30}
\DeclareMathDelimiter{\rrbracket}{5}{matha}{"77}{mathx}{"38}

\newcommand{\defeq}{\coloneqq}
\newcommand{\set}[2]{\{\,#1\mid#2\,\}}
\newcommand{\boplus}{\boldsymbol{\oplus}}
\newcommand{\id}{\mathit{id}}

\newcommand{\sem}[1]{\left\llbracket#1\right\rrbracket}
\newcommand{\UKAddr}{\mathbf{Addr}_{\mathrm{UK}}}
\newcommand{\UKPost}{\mathbf{Post}_{\mathrm{UK}}}
\newcommand{\UKPPost}{\mathbf{PPost}_{\mathrm{UK}}}
\newcommand{\UKPop}{\mathbf{Pop}_{\mathrm{UK}}}

\newcommand{\OPCM}{\mathbf{PCM}_\preceq}

\newcommand{\Pow}{\mathbb{P}}

\newcommand{\op}{\mathrm{op}}

\title{An algebraic theory for data linkage\thanks{This research was supported
    by the EPSRC project \emph{Data Release---Trust, Identity, Privacy and
  Security (EP/ N028139/1 and EP/N027825/1)}.}} \author{Liang-Ting Chen \and
  Markus
Roggenbach \and John V.\ Tucker}
\authorrunning{L.-T.~Chen \and M.\ Roggenbach \and J.V.\ Tucker}
\institute{Department of Computer Science, Computational Foundry, College of
  Science \\ Bay Campus, Swansea University, Swansea SA1 8EN,
UK}

\begin{document}

\maketitle

\begin{abstract}
  There are countless sources of data available to governments, companies, and
  citizens, which can be combined for good or evil. We analyse  the concepts of
  combining data from common sources and linking data from different sources. 
  We model the data and its information content to
  be found in a single source by a partial ordered monoid, and the transfer of
  information between sources by different types of morphisms. To capture the
  linkage between a family of sources, we use a form of Grothendieck
  construction to create a partial ordered monoid that brings together the
  global data of the family in a single structure. We apply our approach to
  database theory and axiomatic structures in approximate reasoning. Thus,
  partial ordered monoids provide a foundation for the algebraic study for
  information gathering in its most primitive form. 
\end{abstract}

\section{Introduction}
There are countless public and private sources of data that can be linked and
analysed for all sorts of reasons, and with all sorts of consequences. The
extraordinary variety of what may be considered data---i.e., data that is
informative in some way---is a challenge to attempts to discover general
principles and techniques for understanding linkage. Motivated by movements for
data sharing we try to uncover general structures common to disparate
situations.

\subsection{Motivation: Exploiting open datasets}

The vast stores of data built up by governments, agencies, institutions and
companies in the course of their operations hold information of value in diverse
and unexpected situations. Some governments have launched initiatives to
encourage bodies to share their data with other organisations and the public.
The released open data is intended to improve transparency, allowing
accountability and engagement with decision making. A systematic review is
\cite{Attard2015}. 

For example, in the UK, there are several national and local registers and a
plethora of statistical data that are now widely shared.
A simple example of the commercial use of open datasets are web services for
selling and letting properties such as Zoopla. In addition to traditional
information about a property, official financial data about local house sales and
crime statistics are provided. 

The UK's Open Data Initiative demonstrates the ambition  to publish internal
government data as open datasets.  There are many patterns of data sharing, of
which three are particularly important:
\begin{enumerate*}
  \item making data public---data release into the wild;
  \item data sharing by contract with a data analysis organisation; and
  \item data sharing with delegation to a new data controller for further onward
    sharing.
\end{enumerate*}
However, data custodians have a legal duty, and a social duty of care,
to ensure that privacy is not breached by the release of open data
sets.

The technical question arises: What information is revealed by, or can be
inferred from, the data?  Naturally, prior to its release, a data set can be
filtered and anonymised but 
\begin{enumerate*}
  \item anonymisation is difficult and often flawed; and
  \item data from various other sources can be combined with a given
    data set to reveal much more.
\end{enumerate*}
There are many data sources to call upon, and many unknown unintended
consequences in making data publicly available.

An early example is Sweeney's finding~\cite{Sweeney1997} that $97\%$
of voters in Cambridge, Massachusetts, USA, can be uniquely identified
by birth dates and postcodes; these can be further
linked with a hospital discharge database to discover individuals'
medical history---e.g., of the governor of Massachusetts at that time~\cite{Sweeney2002a}.

Lately, Narayanan and Shmatikov~\cite{Narayanan2008} devised an
algorithm exploiting sparsity to combine datasets. As a case study
they analysed the Netflix prize dataset and found \textquote{$84\%$
  of (Netflix) subscribers present in the dataset can be uniquely
  identified if the adversary knows six out of eight movies outside
  the top 500} that the subscriber rated.  Such source of film ratings
may come from social engineering or the Internet Movie Database
(IMDb). In response to these privacy concerns, Netflix decided to
withdraw the datasets. Unfortunately, they are still available to
download using BitTorrent or \url{https://archive.org}.

\subsection{Algebraic models of combination and linkage}

In this paper we take a fresh look at the challenge of combining data
sets and linking pieces of data. Our aim is to develop abstract tools to
analyse formally the general nature of data sharing, and technical
issues of policy specification and compliance.  To this end, we
seek algebras of data representations, whose operations combine
two or more pieces of data from the same source to form
data with higher information content. These data representation algebras are to be
defined axiomatically. In its simplest form---that presented here---such an
algebra is an ordered structure with a partial commutative binary operation
$\oplus$ and an identity element~$0$, namely, an \emph{ordered partial commutative
monoid}. 
The operation $\oplus$ \emph{combines} data from the same source. Morphisms
between such monoids model the transfer of data between sources---a process we
call \emph{linkage}.  We create an ordered partial commutative monoid that
brings together all the data from a family of sources using a simplified
Grothendieck construction. We show that our monoid theory of linkage applies to
databases and approximate reasoning. 

\section{Algebras for data combination}

\subsection{Information ordering}\label{sec:info-order}
Data itself is often hierarchical or due to uncertainty becomes so. In this
paper, when we
reason about data, we implicitly work on a set with an ordering that measures
specificity, knowledge, or informativeness. Ideas of information
ordering are nothing new, as they appear to be well-known to different
communities working on uncertainty reasoning~\cite[Section~2.7]{Halpern2003},
multi-valued logic~\cite{Belnap1977}, program
semantics~\cite{Scott1976}, formal concept
analysis~\cite[Chapter~3]{Davey2002a}, and (implicitly) anonymisation
techniques~\cite{Machanavajjhala2007,Wang2014}, to name but a few.

\begin{definition}
  Given a set $X$, an \emph{information order} $\preceq$ on $X$ is a preorder,
  i.e.\
  \begin{enumerate*}
    \item $x \preceq x$ and
    \item $x \preceq y \preceq z$ implies $x \preceq z$.
  \end{enumerate*}
  An \emph{information space} is merely a preordered set $(X, \preceq)$. 
\end{definition}
To illustrate the use of preordered sets in the context of data release and
privacy, we discuss in some detail the use of postcodes to identify locations.

\begin{example} \label{ex:british-postcode}
The taxonomic
hierarchy of British postal codes mostly consists of $6$ to $8$
alphanumeric characters in a format detailed below. Each postcode is divided
into the outward code and the inward code by a single space `\textvisiblespace'.
Each component is formed of further two further parts and each part covers a
smaller area.  For example, \texttt{SA2\textvisiblespace8PP} is the full
postcode of the Singleton Campus of Swansea University and it is understood as follows:
\begin{center}
  \begin{tabular}{| c | c | c | c|}
    \hline
    \texttt{SA} & \texttt{2} & \texttt{8} & \texttt{PP} \\
    \hline
    Postcode Area &	Postcode District	& Postcode Sector	& Postcode Unit \\
    \hline
    \multicolumn{2}{| c |}{Outward Code} & \multicolumn{2}{|c|}{Inward Code} \\
    \hline 
    \multicolumn{4}{|c|}{Postcode} \\
    \hline
  \end{tabular}
\end{center}
Let the set of all full postcodes be denoted by $\UKPost$.

For simplicity, a \emph{partial} postcode refers to a code, where less
signifiant parts might be missing, ordered by prefix order including the empty
string `$\epsilon$' as a special postcode indicating everywhere. For example,
\texttt{SA} stands for Swansea and \texttt{SA2} for a district in Swansea,
and we have partial postcodes 
\[
  \epsilon \preceq \mathtt{SA} \preceq \mathtt{SA2} \preceq
  \texttt{SA2\textvisiblespace8} \preceq \texttt{SA2\textvisiblespace8PP}
\]
note that $\texttt{???\textvisiblespace8PP}$ is not a partial postcode. 
Let us denote the set of all partial postcodes by $\UKPPost$.  

Each full postcode is incomparable with another, as each of them stands for
a disjoint set of postal addresses. On the contrary, the set of partial postcodes
possesses the prefix order $\preceq$ for the hierarchy.  Every partial
postcode $P$ can be realised as a set of full postcodes by
\[
  \sem{P} \defeq \set{ p \in \UKPost }{\text{$P$ is a prefix of $p$}}.
\]
For instance, an empty string $\epsilon$ is realised by $\UKPost$, as it
contains no information apart from being a postcode. Each full postcode $P$ in
$\UKPost$ is realised by the singleton set $\{ P \}$. Note that
$\sem{P}$'s are always non-empty. \qed
\end{example}

The reader may find our definition of information space intriguing. For example,
why is this only a preordered set instead of a partially ordered set? Indeed, as
we can observe from the above example, there are two possible representations of
partial knowledge for postcode:
\begin{enumerate}
  \item $\Pow^+(\UKPost)$---the non-empty powerset of full postcodes.
  \item $\UKPPost$---the set of partial postcodes determined by its format, or
\end{enumerate}
The first representation can be called the \emph{possible world
representation}~\cite[Section~2.1]{Halpern2003} and is well-understood in the
community of knowledge representation and it is more expressive and general.
Every taxonomic hierarchy can be realised by the possible world
interpretation, as each classification level is merely a partition of entities
in a hierarchy. The reverse inclusion order `$\supseteq$' reflects the information
order of taxonomic hierarchy, i.e.\ $P$ is of higher hierarchy than $Q$ only if
$\sem{P} \subseteq \sem{Q}$ and `$\supseteq$' is surely a partial order. 
We return to this general points in \Cref{sec:possibility}.

On the other hand, the second kind of representations is often what we have in
the first place or what we would like to use in data release. The information order
$\preceq$ requires some effort to decide, but generally it is clear from the
context. However, we may have two different representations for the very same
set of entities.  If a weight is attached to the data in question, then the
second representation is more manageable than the first:

\begin{example}
  Consider a version due to a privacy concern.\footnote{%
    Some privacy protection models are achieved by generalisation and
    suppression of cell values, see~\cite{Sweeney2002a} for example.}
  \begin{figure}
    \begin{subfigure}{.48\textwidth}
      \centering
    \begin{tabular}{| c | c |} 
      \hline
      User ID & Postcode \\
      \hline
      1 & \texttt{SA2\textvisiblespace8PP} \\
      \hline
      2 & \texttt{SA2\textvisiblespace8PW} \\
      \hline
      3 & \texttt{SA1\textvisiblespace3LP} \\
      \hline
      4 & \texttt{SA2\textvisiblespace8QF} \\
      \hline
    \end{tabular}
    \caption{Original dataset}
    \label{fig:dataset-postcode-a}
  \end{subfigure}
  \begin{subfigure}{.48\textwidth}
    \centering
    \begin{tabular}{| c | c |} 
      \hline
      User ID & Postcode \\
      \hline
      * & \texttt{SA2\textvisiblespace8} \\
      \hline
      * & \texttt{SA2\textvisiblespace8} \\
      \hline
      * & \texttt{SA1\textvisiblespace3} \\
      \hline
      * & \texttt{SA2\textvisiblespace8} \\
      \hline
    \end{tabular}
    \caption{Sanitised dataset}
    \label{fig:dataset-postcode-b}
  \end{subfigure}
    \caption{Datasets containing postal information}
    \label{fig:dataset-postcode}
  \end{figure}%
  Both build frequency distribution, and some probabilities can be calculated
  based on the information order over postcodes, say, $\Pr[\mathtt{SA2} \preceq
  X]$. 

  In Kolmogorov's probability theory, the first step is to find out a sample
  space $\Omega$ and a $\sigma$-algebra $\Sigma$, and the typical choice is
  $\Omega = \UKPost$ and $\Sigma = \Pow(\UKPost)$. The probability measure
  for the original dataset (\cref{fig:dataset-postcode-a}) is clear. But, it is
  tricky to define faithfully a probability measure for the sanitised dataset
  (\cref{fig:dataset-postcode-b}), since it requires to assign a probability to
  each full postcode with the prefix \texttt{SA1\textvisiblespace3}. The
  convention is to apply the principle of indifference---each postcode of
  $\sem{\texttt{SA1\textvisiblespace3}}$ has the same probability $1/k$ where
  $k$ is the possibly \emph{unknown} number of postcodes in
  $\sem{\texttt{SA1\textvisiblespace3}}$. Even if $k$ is known, the presumed
  probability $1/k$ is an over-approximation of the given information. 

  On the other hand, no matter what probability is assigned to subsets of
  full postcodes, the probability of $\Pr[ \mathtt{SA2} \preceq X ]$ is always
  the sum 
  \[
    \sum_{\mathtt{SA2} \preceq Q} \Pr[X = Q] = 3/4
  \]
  without knowing any further information. The expressiveness is limited if we
  confine ourselves to probabilities of partial postcodes only, since partial
  postcodes are not closed under Boolean connectives contrary to the subset
  representation. This trade-off enables us to represent the
  \emph{exact} information of data. \qed
\end{example}

Another problem of the possible world representation arises if the
information order is by nature \emph{not} anti-symmetric. It is intuitive to see
that \cref{fig:dataset-postcode-a} is more informative than
\cref{fig:dataset-postcode-b}. There are at least three applicable
orderings over subsets $P, Q$ of elements in an information space $X$, which are
\begin{align*}
  P \preceq^\flat Q \iff & \forall x \in P.\,\exists y \in Q.\, x \preceq y\\
  P \preceq^\sharp Q \iff & \forall y \in Q.\,\exists x \in P.\, x \preceq y\\
  P \preceq^\natural Q \iff & P \preceq^\flat Q \land P \preceq^\sharp Q
\end{align*}
The ordering can model a number of processes or situations. $P \preceq^\flat Q$
models that everything in $P$ has a more informative datum in $Q$. So $Q$ is an
enrichment of $P$. Conversely, $P \preceq^\sharp Q$ models that everything in $Q$ 
has a less informative datum in $P$, so $P$ is an adulteration of $Q$. 

Each of the orderings plays a role in various contexts, such as non-deterministic
computation~\cite{Gunter1992} and relative likelihood~\cite[Section
2.7]{Halpern2003}. These orderings are preorders but not
anti-symmetric in general. 

\begin{example}
  Ignoring user ID and repetitions, we have two sets representing the
  information in \cref{fig:dataset-postcode}:
  \begin{align*}
    P_1 & \defeq \{
      \texttt{SA2\textvisiblespace8PP}, 
      \texttt{SA2\textvisiblespace8PW}, 
      \texttt{SA1\textvisiblespace3LP},
      \texttt{SA2\textvisiblespace8QF}\} \\
      P_2 & \defeq \{ 
      \texttt{SA2\textvisiblespace8}, 
      \texttt{SA1\textvisiblespace3}
    \}
  \end{align*}
  The set $P_1$ is more informative than $P_2$ with respect to $\preceq^\flat$, 
  $\preceq^\sharp$, and $\preceq^\natural$. 
  \qed
\end{example}
Even further, the standard equality `$=$' on the data in $X$ is irrelevant from the
information-theoretic perspective, as we only care about the information content
of data. For example, any subset $P$ of an information space $(X, \preceq)$ is
indistinguishable from but fails to be equal to its \emph{convex
hull}\footnote{See, e.g.,~\cite[p.63]{Davey2002a}.} $\mathcal{K}(P) \defeq
\set{ a \in X }{ \exists x, y \in P.\, x \preceq a \preceq y }$, i.e.\ 
\[
  P \preceq^\natural \mathcal{K}(P) \preceq^\natural P \quad\text{but generally}\quad P
  \neq \mathcal{K}(P).
\]


So, we introduce:
\begin{definition}
  Given an information order $\preceq$ on a set $X$, define an equivalence
  relation by
  \[
    {x \cong y} \iff x \preceq y \text{ and } y \preceq x
  \]
  and $x$ is said to be \emph{equivalent} to $y$. Each element in the same
  equivalence class is of the same \emph{information content}.
\end{definition}

From a mathematical viewpoint, each element $x$ is a representative of the
information class $[x]$. Every representative of the same class embodies the
same amount of information with respect to the information order $\preceq$.
Computing and deciding the information class could be costly and conceptually
gain little, so it is easier to work and present our latter formulations
with representatives directly. 

\begin{remark} \label{re:information-equivalence}
  From this, we can argue further that `$\cong$' is the right notion of
  equality where the strict equality `$=$` plays no role at all in
  an ordered setting.  Indeed, the convention is to consider the quotient
  $(X/{\cong}, {\preceq}/{\cong})$ as the poset of information and $[x] = [y]$
  is equivalent to $x \cong y$, but this convention makes notations rather
  heavy.
\end{remark}

So the point is that only the preorder $\preceq$ for information matters and it
fails to be a partial order in general. 

\subsection{Ordered partial commutative monoids} \label{sec:opcm}
To combine and link data across various domains yields data that is presumably more informative
than the separate pieces of information alone. In this section, we introduce an
algebraic operation over an information space for combining data. 
Central to our investigation is the concept of ordered partial commutative
monoids. Whilst monoids of many kinds, e.g., ordered commutative
monoids~\cite{Fritz2015} and partial commutative
monoids~\cite{Wehrung2015,Foulis1994}, have been discovered and developed in
many application areas, surprisingly we have not found a monoid combining
both---ordering and partiality. A possible exception we found is monoids
viewed as a degenerated class of partial monoidal categories defined in
\cite{Coecke2010}. 

\begin{definition} \label{def:opcm}
  An \emph{ordered partial commutative monoid} ${(M, \preceq, \oplus,
  0)}$ consists of
  \begin{enumerate*}
    \item a preordered set $(M, \preceq)$, 
    \item a constant $0 \in M$, and
    \item a partial binary operation
      $\oplus\colon M \times M \rightharpoonup M$, i.e.\ 
      $x \oplus y$ may not be defined. 
  \end{enumerate*}
  For brevity, `$x \perp y$' stands for `$x \oplus y$' is defined. Further, $(M,
  \preceq, \oplus, 0)$ satisfies the properties below.
  \begin{enumerate}[label=(OPCM{\arabic*}),align=left,ref=(OPCM{\arabic*})]
    \item \label{def:OPCM1} $0 \oplus x \cong x$.

    \item \label{def:OPCM2} $y \perp x$ and $x \oplus y \cong y \oplus x$
      if $x \perp y$.

    \item \label{def:OPCM3} $x \perp y$, $(x \oplus y) \perp z$, and $x \oplus (y
      \oplus z) \cong (x \oplus y) \oplus z$ if $y \perp z$ and $x
      \perp (y \oplus z)$.

    \item \label{def:OPCM4} $x_{1} \oplus y \preceq x_{2} \oplus y$ if $x_{i}
      \perp y$ for $i = 1, 2$ and $x_{1}\preceq x_{2}$.
    \end{enumerate}
  An ordered partial commutative monoid is written as \emph{OPCM} for short.  An
  (unordered) \emph{partial commutative monoid} $(M, \oplus, 0)$, PCM for short,
  is an OPCM with the discrete ordering $x \preceq y \iff x = y$. An
    \emph{ordered commutative monoid} is an OPCM with the binary operation
    $\oplus$ being total. 
\end{definition}

The element $x \oplus y$ denotes data that represents a combination of the
information of $x$ and $y$. The constant $0$ stands for some vacuous information
so that $x \oplus 0$ is always defined and equivalent to $x$. 

Referring to \Cref{re:information-equivalence}, the following fact shows that
the use of `$\cong$' is equivalent to the standard equality `$=$' in the
partially ordered quotient:
\begin{proposition}
  Let $(M, \preceq, \oplus, 0)$ be an OPCM. Then, 
  \begin{enumerate}
    \item the relation defined by $[x] \leq [y] \iff x \preceq y$ on the
        quotient set $\sfrac{M}{\cong}$ is a partial order and $[x] = [y] \iff x \cong
          y$;
        \item $(\sfrac{M}{\cong}, \leq, [\oplus], [0])$ with $[x]
          \mathbin{[\oplus]} [y]$ defined as $[x \oplus y]$ is an OPCM.
  \end{enumerate}
\end{proposition}
\begin{proof}[Sketch]
  The first statement is well-known and obvious. By 
  \ref{def:OPCM4} the proof of the second
  statement is straightforward. 
\end{proof}

Partiality enables us to encapsulate consistency or other premisses. That is,
$x$ may contradict $y$ so that no viable information can be
derived.


The algebraic structure of a PCM also gives rise to a natural ordering between
information purely determined by the combination $\oplus$. 
\begin{definition}
  The \emph{algebraic ordering} on an OPCM is defined by
  \[
    x \sqsubseteq  y \iff \exists z.\, x \oplus z \cong y. 
  \]
\end{definition}
\begin{proposition}\label{prop:algebraic-pcm}
  Every PCM $(M, \sqsubseteq, \oplus, 0)$ with algebraic ordering is an
  \begin{enumerate}
    \item OPCM which satisfies 
    \item $0 \sqsubseteq x$, and that
    \item if $(x, y) \sqsubseteq (x', y')$, $x' \perp y'$, $x \perp x$, then
      $x \oplus y \sqsubseteq x' \oplus y'$. 
  \end{enumerate}
\end{proposition}

The algebraic ordering of an OPCM $(M, \preceq, \oplus, 0)$ is compatible
with the information ordering if the identity $0$ is the $\preceq$-least
informative element:

\begin{proposition}
  Let $(M, \preceq, \oplus, 0)$ be an OPCM such that $0 \preceq x$. Then, 
  \begin{enumerate}
    \item $x \sqsubseteq y \implies x \preceq y$;
    \item $x, y \preceq x \oplus y$ whenever $x \perp y$. 
  \end{enumerate}
\end{proposition}

\begin{proof}
  The assumption is equivalent to $x \oplus z \cong y$ for some $z$ and $0
  \preceq z$. The monotonicity of $\oplus$ shows that $x \cong x \oplus 0 \preceq x
  \oplus z \cong y$, i.e.\ $x \preceq y$. 

  Similarly, $x \cong x \oplus 0 \preceq x \oplus y$ by monotonicity. 
\end{proof}


\begin{remark}
  The implication along with \Cref{prop:algebraic-pcm} suggests the condition $0
  \preceq x$ is decisive, otherwise $\oplus$ may not represent `combination of
  information' but something different (cf.\ the semantics of Belnap's
  $4$-valued logic~\cite{Arieli1998}).
  
\end{remark}

\begin{example} \label{ex:postcode-pcm}
  Consider the collection of all non-empty subsets of full
  postcodes~$\Pow^+(\UKPost)$ equipped with the reverse inclusion order $P_1
  \preceq P_2$ iff $P_2 \subseteq P_1$. The intersection $\cap$ of subsets as a
  combination operation $\oplus$, is a partial operation, since $P_1 \cap P_2$
  might be empty and $\empty \not\in \Pow^+(\UKPost)$. Clearly, the
  intersection is monotone with respect to the reverse inclusion order.
  Similarly, the set of partial postcodes equipped with the prefix
  ordering~$\preceq$ discussed in~\Cref{ex:british-postcode} has a simple OPCM
  structure: $x \oplus y$ is defined as $\max\{x, y\}$. 
\end{example}

\subsection{Homomorphisms} \label{sec:homomorphism}

The internal structure of an OPCM models data and information of a single
source. So the external interaction between OPCMs models a comparison,
combination, interpretation, or linkage between sources. Various kinds of
structure preserving maps between OPCMs arise naturally, e.g.,\ order-preserving
maps, $\oplus$-preserving maps, or both. We begin with the familiar one. 

\begin{definition} \label{def:homomorphism}
  A \emph{homomorphism} $M \xrightarrow{f} N$ of OPCMs is a
  function satisfying 
  \begin{enumerate}[label=(HOM{\arabic*}),align=left,ref=(HOM{\arabic*})]
    \item $x \preceq_M y \implies fx \preceq_N fy$ 
    \item \label{def:homomorphism-2} $f(0_M) \cong 0_N$
    \item \label{def:homomorphism-3} $x \perp y \implies f(x \oplus_M y) \cong fx
      \oplus_N fy$
  \end{enumerate}
  The collection of OPCMs with their homomorphisms forms a category $\OPCM$. 
\end{definition}

An `interpretation' of information in a different domain of discourse or
context, is a typical example of a homomorphism. The
\emph{trivial map} $f\colon M \to N$ defined by $f(x) = 0$ is a homomorphism that
destroys all the information in $M$. The set of
partial postcodes \emph{per se} is merely a set of strings following specific
format, so it makes little sense to say how rare a postcode $P$ is among other
postcodes; it becomes meaningful when it refers to certain geographic area,
population, or other associated information.
\begin{example} \label{ex:UKPost->UKPop}
  Let $\UKPop$ denote the UK population. Assume that
  \begin{enumerate*}
    \item everyone (of interest) is registered with exactly one postcode for
      their main residence, and
    \item each postcode is associated with someone. 
  \end{enumerate*}
   The assumption amounts to a surjective function $f\colon \UKPop \to
   \UKPost$. 
  
  Consider the possible world representation for $\UKPop$. Each set $S$ of
  postcodes then can be interpreted as the set $\sem{S} \defeq f^{-1}(S)
  \subseteq \UKPop$ of population officially registered in the area specified by
  $P$. The mapping $\sem{-}\colon \Pow^+\UKPost \to \Pow^+\UKPop$ is clearly
  homomorphic w.r.t.\ the OPCM discussed in \Cref{ex:postcode-pcm}, since 
  \begin{enumerate}
    \item it is monotone, as $\sem{S_1} \supseteq \sem{S_2}$ if $S_1 \supseteq
      S_2$;
    \item it preserves the identity, as $f^{-1}(\UKPost) = \UKPop$;
    \item and moreover $\sem{S_1 \cap S_2} = \sem{S_1} \cap \sem{S_2}$ as
      $f^{-1}$ preserves intersection. 
  \end{enumerate}
  \qed
\end{example}

Besides concrete homomorphisms, one has the following standard
notions: \emph{isomorphism}, \emph{monomorphism}, \emph{embedding},
\emph{epimorphism}, and so on, following the doctrine of category theory. Among
them, the product of two OPCMs can be understood as pairs of independent
sources of information. 

\begin{definition}
  The \emph{product monoid} $M_1 \times M_2$ of $M_i = (M_i,
  \preceq_i, \oplus_i, 0_i)$ for $i = 1, 2$ is the cartesian product equipped with
  \begin{enumerate}
    \item the pointwise ordering $(x_1, x_2) \preceq (y_1, y_2) \iff x_1
        \preceq_1
        y_1 \land x_2 \preceq_2 y_2$, 
    \item $0 \defeq (0_1, 0_2)$, and 
    \item $(x_1, x_2) \oplus (y_1, y_2) \defeq (x_1 \oplus_1 y_1, x_2 \oplus_2
      y_2)$ if $x_1 \perp y_1$ and $x_2 \perp y_2$. 
  \end{enumerate}
\end{definition}
The universal property for product shows that $M_1 \times M_2$ consists of pairs
of independent pieces of information from $M_1$ and $M_2$:
\begin{proposition}
  For any OPCM $N$ and any pair of homomorphisms $f_i\colon N \to M_i$
  for $i = 1, 2$, there exists a unique homomorphism $h\colon N \to M_1 \times
  M_2$ such that $\pi_i \circ h = f_i$, where $\pi_i$ is the $i$-th projection
  homomorphism. 
\end{proposition}

Another useful notions are embedding and isomorphism. 
\begin{definition}
  A homomorphism $e\colon M \rightarrowtail N$ is an \emph{order-embedding} if
  it not only preserves but also reflects the ordering:
    $e(x) \preceq e(y) \iff x \preceq y$.
      An \emph{isomorphism} is a bijective order-embedding. 
\end{definition}

\section{Further examples}

\subsection{Flat algebras}
The most simple OPCM is perhaps a set $X$ equipped with an additional
element $\bot$ denoting \emph{unknown} and $x \leq y$ iff $x = \bot$ or $x = y$
with $x \oplus y \defeq \text{(the join of $x$ and $y$)}$.
In spite of its simplicity, it has been elaborated further in relational database
theory~\cite[Chapter~8]{Codd1990}.

\subsection{Possibilities over a set} \label{sec:possibility}
We have used a possible world representation discussing postcodes
(\Cref{sec:opcm}). In this
section, we study its general properties. As the reader may have observed from
our examples about non-empty subsets of full postcodes, the argument is
completely generic and can be applied to any non-empty set $X$. In short, we
have the following generalisation of \Cref{ex:postcode-pcm}:
\begin{proposition}
  For any non-empty set $X$, the non-empty powerset $\Pow^+X$ with the reverse
  inclusion and intersection forms an OPCM $(\Pow^+X, \supseteq, \cap,
  X)$. 
\end{proposition}
In general, the set $X$ represents some elementary form of atomic information
such as \emph{codes}, \emph{labels}, \emph{tags} or \emph{facts} from which is
made. The data in the source is a non-empty subset
$S$ of $X$ containing a set of possible choices from
$X$.

\subsection{Possibilities over an OPCM}
It is often the case that only pieces of information shared by a group of people is known
instead of each individual's.  As each piece of information
in our algebraic theory is an element of some OPCM, we proceed with non-empty
subsets of an OPCM which is in turn another OPCM. 

The starting point is the observation that a mere intersection of two subsets of
an OPCM $(M, \preceq, \oplus, 0)$ would exclude combinable but not exactly the
same information.  Note that we can reformulate intersection in a rather
silly way as
\[
  P \cap Q = \set{ x }{ x \in P, y \in Q, x = y }
\]
We can utilise `$\oplus$' and define a combination of two subsets of OPCM by
\[
  P \boplus Q \defeq \set{ x \oplus y }{ x \in P, y \in Q, x \perp y } 
\]
consisting of refined information only. How about the information order between
subsets?  It turns out that only one of orderings for powerset introduced in
\Cref{sec:info-order},
\begin{align*}
  P \preceq^\sharp Q \iff & \forall y \in Q.\,\exists x \in P.\, x \preceq y
\end{align*}
is a sensible preorder with respect to the definition of $P\oplus Q$.

\begin{theorem}\label{thm:possibility}
  Let $(M, \preceq, \oplus, 0)$ be an OPCM such that $M$ is
  $\oplus$-downward closed, i.e.\ if $x \preceq x'$ and $x' \perp y$ then $x
  \perp y$. For non-empty subsets $P$ and $Q$,
  \[
    P \boplus Q \defeq \set{ x \oplus y }{ x \in P, y \in Q, x \perp y }.
  \]
  Then, 
  \begin{enumerate}
    \item $\Pow^+M = (\Pow^+M, \preceq^\sharp, \boplus, \{0\})$ is also an
      OPCM; \item $\{0\} \preceq^\sharp P$ for any $P$ if $0 \preceq x$
      for any $x \in
      M$.
  \end{enumerate}
\end{theorem}
\begin{proof}
  $P \oplus Q$ is defined if there is at least a pair of $x \in P$ and $y \in
  Q$ with $x \oplus y$ defined. Given a non-empty subset $P \subseteq M$, it is
  clear that $P \oplus \{0\} \cong P$ by definition. Let $Q$ be another
  non-empty subset.  Then, $P \oplus Q \cong Q \oplus P$ since $x \oplus y \cong
  y \oplus x$ for any $x \in P$ and $y \in Q$ with $x \perp y$. Similarly, $P
  \oplus (Q \oplus R) \cong (P \oplus Q) \oplus R$.  As for monotonicity of
  $\oplus$, assume $P_1 \preceq^\sharp P_2$. Then, for any $x_2 \in P_2$ and $y
  \in Q$ such that $x_2 \perp y$, there exists $x_1 \in P_1$ with $x_1 \preceq
  x_2$ by $P_1 \preceq^\sharp P_2$ and $x_1 \perp y$ by assumption. By
  monotonicity, $x_1 \oplus y \preceq x_2 \oplus y$. Therefore, $P_1 \oplus Q
  \preceq^\sharp P_2 \oplus Q$. 

  The second statement follows from definition and assumption directly. 
\end{proof}

\section{Data linkage}

A domain of discourse can have a number of data sources so that the same piece
of information can be understood in various contexts differently.  How do we
know that the original information remains intact?

\subsection{Change of domain}
A homomorphism $f\colon M \to N$ qualifies as a mapping changing domains from
$M$ to $N$ but it can lose data, e.g.\ the trivial map $f(x) = 0$ destroys all
data. One way to avoid this problem is to use homomorphisms with a
restriction map $f^*\colon N \to M$ satisfying a `preservation condition' $x
\leq f^* f(x)$ for $x \in M$. 

\begin{definition}
  A homomorphism $f\colon M \to N$ is a \emph{change of domain} if $f$ is a
  lower adjoint,\footnote{%
    Every adjoint is unique up to order isomorphism---that is, if $g$ is an
    upper adjoint of $f$ then $f^*y \cong g y$ for any $y$, so we can say that a
    homomorphism $f$ is a change of domain without referring to~$f^*$.} i.e.\
    there exists an order-preserving map $f^*\colon N \to M$ such
  that
  \[
    fx \preceq_N y \iff x \preceq_M f^*y
  \]
\end{definition}
Our formal definition requires that an extension $f$ with its
restriction $f^*$ forms a \emph{Galois connection}~\cite{Davey2002a}.

Every Galois connection $(f, f^*)$ gives rise to a \emph{closure operator}---a
monotone function $f^* \circ f$ satisfying
\begin{enumerate*}
  \item $x \preceq f^*f(x)$ and
  \item $f^*f(f^*f x)\preceq f^* fx$. 
\end{enumerate*} 
Intuitively, the information represented by $f^*f(x)$ is at least as informative
as $x$. 

The class of changes of domain is closed under composition. It is not hard to
see that the composite $k\circ f$ of two lower adjoints is again a lower
adjoint, because $k \circ f$ is homomorphic and by definition
\[
  k(fx) \preceq z \iff fx \preceq k^*z \iff x \preceq f^*k^*z.
\]
Trivially, an identity function $\mathit{id}$ is itself a change of domain.
Therefore, the class of OPCMs and changes of domain forms a subcategory
of $\OPCM$. 
\begin{example}
  The homomorphism $\sem{-}\colon \Pow^+\UKPost \to \Pow^+\UKPop$ discussed in
  \Cref{ex:UKPost->UKPop} is indeed a change of domain. The restriction from
  $\Pow^+\UKPop$ to $\Pow^+\UKPost$ is given by mapping a set of population to
  the set of their registered postcodes. The existence of this restriction
  follows from the assumption that everyone of interest signs a register with a
  full postcode. Formally, the restriction is the forward-image function
  of the surjection $f\colon \UKPop \to \UKPost$ given by our assumption, so 
  \[
    \sem{S} \preceq A \iff f^{-1}(S) \supseteq A \iff S \supseteq f[A]
      \iff S \preceq f[A] 
  \]
  for any non-empty $S \subseteq \UKPost$ and $A \subseteq \UKPop$. 
\end{example}

Given a change of domain $f\colon M \to N$, there are two different ways to
combine $x \in M$ with $y \in N$. Their relationship can be stated as follows:
\begin{proposition} \label{prop:extened-oplus}
  Given a change of domain $f\colon M \to N$, the following 
  \[
    x \oplus f^*y \preceq f^*(f x \oplus y)
  \]
  always holds for any $x \in M$ and $y \in N$. 
\end{proposition}
\begin{proof}
  By $f^*y \preceq f^*y \iff f f^*y \preceq y$, it follows that
  \[
    f (x \oplus f^*y) * \cong f x \oplus ff^*y \preceq f x \oplus
    y 
    \iff x \oplus f^* y \preceq f^*(f x \oplus y). 
    \]
\end{proof}

Armed with these notions, we now formally define `linkage' as follows.
\begin{definition}
  A \emph{linking passage} $(f_i, g_i)_{i = 1, 2}$ of $M_1$ and $M_2$ is a commutative
  diagram of changes of domain up to equivalence:
  \[
    \xymatrix@l@-.5em{
      & K \ar[ld]_{g_2} \ar[rd]^{g_1} \\
      M_2 \ar@{}[rr]|{\cong}\ar[rd]_{f_2}  & & M_1 \ar[ld]^{f_1} \\
          & N }
  \]
  i.e.\ the equation $f_1 \circ g_1(k) \cong f_2 \circ g_2(k)$ for any $k \in K$.
  Given a linking passage as above, elements $x_i \in M_i$ can be 
  \emph{linked} as $\bigoplus_{i} f_ix_i$ in $N$.
\end{definition}
In the context of information, the OPCM $K$ above is some \emph{common
domain} of discourse between $M_1$ and $M_2$; $N$ is some domain at least including
$M_1$ and $M_2$.

Given a linking passage of $M_1$ and $M_2$, there are two ways
transferring information from $M_1$ to $M_2$---one through the larger domain $N$
and the other through their common domain $K$. The former route intuitively
preserves more information than the other, and this intuition can be 
justified as follows.
\begin{proposition} \label{prop:two-extensions}
  Given a linking passage $(f_i, g_i)_{i = 1, 2}$ and for any $x \in M_1$,
  the inequation 
  $f_2^* \circ f_1(x) \preceq g_2 \circ g_1^*(x)$ holds. Diagrammatically, 
  
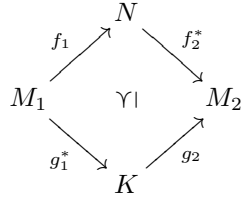
\[
    \xymatrix@-.5em{
      & N \ar[rd]^{f_2^*} \ar@{}[dd]|{\rotatebox[origin=c]{90}{$\preceq$}}\\
      M_1 \ar[ru]^{f_1} \ar[rd]_{g_1^*} & & M_2 \\
                                        & K \ar[ru]_{g_2}
    }
  \]
\end{proposition}
\begin{proof}
  By $f^*y \preceq f^*y \iff ff^*y \preceq y$, it follows that
  \[
    k(y) \succeq k \circ f\circ f^*(y) \succeq g \circ h \circ f^*(y)
    \iff g^* \circ k(y) \succeq h \circ f^*(y).
  \]
\end{proof}
\begin{example}
  Assume that $M_1 = \Pow^+(X \times Y)$ and $M_2 = \Pow^+(Y \times Z)$. 
  Then, 
  
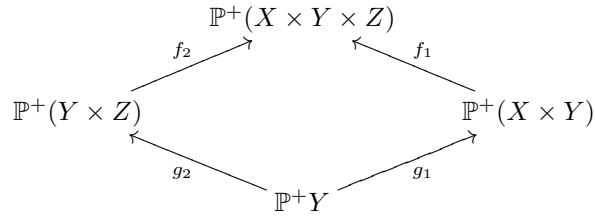
\[
    \xymatrix@-.5em@l{
      & \Pow^+Y \ar[ld]_{g_1} \ar[rd]^{g_2} \\
      \Pow^+(X \times Y) \ar[rd]_{f_1}  & & \Pow^+(Y \times Z) \ar[ld]^{f_2} \\
                                        & \Pow^+(X \times Y \times Z) }
  \]
  is a linking passage where $f_1, f_2, g_1, g_2$ are preimage
  functions of corresponding projections. Moreover, the subset $f_1(U) \cap
  f_2(V)$ is equal to 
  \[
    \set{ (x, y, z)}{ (x, y) \in U \land (y, z) \in V }
  \]
  for any non-empty $U \subseteq X \times Y$ and $V \subseteq Y \times Z$, which
  is the \emph{natural join} in relational database theory.
  For a plausible  example in practice, consider $U \subseteq \UKPop \times
  \UKAddr$ a non-empty set of suspects with their hiding places and $V \subseteq
  \UKAddr \times \UKPop$ a non-empty set of house addresses and their owners.
  The combined information $f_1(U) \cap f_2(V)$ may represent pairs of suspects,
  addresses, and house owners who possibly provide shelters to suspects. 
\end{example}

\subsubsection{Local computation scheme}
In practice, each datum $x_i$ about the attribute $i$ is collected from
various data sources $M_i$. To combine every $x_i$'s, we can combine them in a
common domain $M$ and then restrict the combined information to a smaller domain
$N$ of interest, i.e.\
\[
  g^*\left(\bigoplus_{i = 1}^n f_ix_i\right)
\]
represented symbolically. 
The computation is usually costly, however. One interesting observation stated
as \emph{the combination axiom} from~\cite{Kohlas2003} in a similar form is that
the above information can be computed locally without the need of extending
everything to $M$ if inequalities in \Cref{prop:two-extensions,%
prop:extened-oplus} are in fact equivalences for involved changes of domains.
This observation is useful for developing an efficient computation algorithm
that is, however, beyond the scope of this paper. 

\subsection{Possibilities over a set}

A surjective function $X \twoheadrightarrow Y$ gives rise to a change of domain
from $\Pow^+Y$ to $\Pow^+X$. The surjectivity requirement is essential to ensure
that a non-empty subset $S \subseteq Y$ is mapped to a non-empty subset
$f^{-1}(S) \subseteq X$. 
\begin{proposition}
  For any surjective function $f\colon X \twoheadrightarrow Y$, there is a Galois connection
  \[
    f^{-1}(V) \supseteq U \iff V \supseteq f[U]
  \]
  where the preimage function $f^{-1}$ is a homomorphism from $\Pow^+ Y$ to
  $\Pow^+X$ and the forward-image function $f[{-}]\colon \Pow^+ X \to \Pow^+Y$
  is monotonic. 
\end{proposition}

It is straightforward to see that the inequality of \Cref{prop:extened-oplus} is
an equality for any change of domain given by a surjective function. That is, 
\[
  f[f^{-1}(U) \cap V] = U \cap f(V) 
\]
for any $U$ and $V$ by simple calculations. 

\begin{proposition}
  Suppose that there are $f_i\colon X_i \twoheadrightarrow Z$ and $g_i\colon Y
  \twoheadrightarrow X_i$  for $i = 1, 2$ with $g_1 \circ f_1 = g_2 \circ f_2$.
  Then, $(g_i^{-1}, f_i^{-1})$ is a linking passage, i.e.\, 
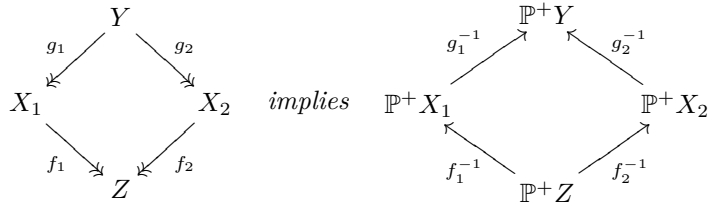
  \[
    \vcenter{
      \xymatrix@-.5em{
        & Y \ar@{->>}[ld]_{g_1} \ar@{->>}[rd]^{g_2}\\
        X_1 \ar@{->>}[rd]_{f_1}  & & X_2 \ar@{->>}[ld]^{f_2} \\
                                          & Z 
      }
    }
    \quad\text{implies}\quad
    \vcenter{
      \xymatrix@-.5em{
        & \Pow^+Y \\
        \Pow^+X_1 \ar[ru]^{g_1^{-1}}  & & \Pow^+X_2 \ar[lu]_{g_2^{-1}} \\
                                 & \Pow^+Z \ar[ru]_{f_2^{-1}} \ar[lu]^{f_1^{-1}}
      }
    }
  \]
\end{proposition}

If a linking passage is created by functions $g_i\colon Y \to X_i$, then
non-empty subsets $U_i \subseteq X_i$  can be linked as a subset of $Y$
\[
  U_1 \boxplus U_2 =  g^{-1}_1(U_1) \cap g^{-1}_2(U_2).
\]

\begin{example}
  Let $A$ be a set of \emph{attributes} and for each $a \in A$ a set $\Phi_a$ a
  set of values for the attribute $i$. For example, $i$ can be $g$ for
  \textquote{gender}, $p$ for \textquote{British postcode}, $s$ for
  \textquote{salary}, etc., then $\Phi_g$ could be the two-element set $\{
  \male, \female \}$, $\Phi_p = \UKPost$ the set of all full British postcodes,
  and $\Phi_s = \mathbb{N}$ the set of natural numbers.  Given any two sets $I,
  J \subseteq A$ of attributes, we have a commutative diagram 
  \[
    \xymatrix@C+2em{
      \prod_{k \in I \cup J} \Phi_k \ar@{->>}[r]^{g_1} \ar@{->>}[d]_{g_2} &
      \prod_{i \in I} \Phi_i
      \ar@{->>}[d]^{f_1}\\ \prod_{j \in J} \Phi_j \ar@{->>}[r]_{f_2} & \prod_{l \in
      I \cap J} \Phi_l
    }
  \]
  where $g_i$ and $f_i$ are all projections.
\end{example}

\section{Data sources and linkage}

So far, an OPCM $M$ is an abstract collection of data from a data source for a
single domain of discourse that can be combined and compared. A model of data
linkage requires a family of PCMs $(M_i, \preceq_i, \oplus_i, 0_i)$, for $i \in
I$, and ways to move among various sources and domains.  Further, the nature of
sources and domains
induces a structure to the index set $I$, typically an ordering $\preceq$, that
reflects the relationship between sources and domains such as $i \preceq j$. With these
components, we will model and illustrate data linkage using a form of
Grothendieck construction for $I$-indexed OPCMs. 

We will develop the theory in two steps and compare our construction with
axiomatic frameworks in the community of approximate reasoning such as ordered
valuation algebras~\cite{Haenni2004} and information
algebras~\cite{Kohlas2003,Kohlas2012}.

\subsection{Grothendieck construction for preordered sets}

Let $I$ be a preordered set and $P$ an $I$-indexed family of preordered sets
$P_i$ for $i \in I$ together with order-preserving functions $P^i_j\colon P_i
\to P_j$ whenever $i \preceq j$ satisfying 
\begin{enumerate}
  \item $P^i_i \cong \id_{P_i} \colon P_i \to P_i$ is the identity function, and 
  \item $P^j_k \circ P^i_j \cong P^i_k \colon P_i \to P_k$
\end{enumerate}
where $P^j_k \circ P^i_j \cong P^i_k$ means $P^j_k \circ P^i_j(x) \cong
P^i_k(x)$ for every $x$ and similarly for $P^i_i \cong \id$. Note that $P$ is a
\emph{pseudo-functor}. If the above equations hold strictly, then $P$ is a
(proper) \emph{functor}.

\begin{definition}
  The \emph{Grothendieck completion}  of $P$ consists of
  \[
    \int P \defeq \set{ (i, x) }{ x \in P_i }
  \]
  with a relation defined by
  \[
    (i, x) \preceq (j, y) \iff i \preceq j \text{ and } P^i_{j}(x) \preceq y
      \text{ for $x \in P_i$ and $y \in P_j$}
  \]
\end{definition}

The ordering appears natural in our context: $P^i_j$ is
typically a change of domain, and $P^i_j(x)$ is merely an extension of $x$
and $(i, x) \preceq (j, y)$ if and only if $j$ is a larger domain of discourse
than $i$ and the extended form of $x$ is still less informative than $y$. 
\begin{proposition} \label{prop:grothendieck-1}
  The following statements are true: 
  \begin{enumerate}
    \item The above Grothendieck completion $\int P$ is a preordered set. 

    \item If $(I, \leq)$ and every $(P_i, \leq)$ is partially ordered,
      then so is $(\int P, \preceq)$. 

    \item The projection $p\colon \int P \to (I, \preceq)$ is functorial. 

    \item $p$ is an \emph{opfibration}. That is, for every $(i,
      x) \in \int P$, $j$ with $i \preceq j$ there exists $(j, y)$ such that
      $(i, x) \preceq (j, y)$ and moreover for any $(k, z)$ with $(i, x) \preceq
      (k, z)$ and $j \preceq k$ it is also true that $(j, y) \preceq (k, z)$.

    \item If each $P^i_j$ has a right adjoint, then $p$ is an
      \emph{bifibration}, i.e.\ $p^\op\colon (\int P, \succeq) \to (I, \succeq)$
      is also an opfibration. 
  \end{enumerate}
\end{proposition}
\begin{proof}
  \begin{enumerate}
    \item We show that $(\int P, \preceq)$ is a preordered set as follows. 
      \begin{enumerate}
        \item Reflexivity: $(i, x) \preceq (i, x)$ since $P_{ii}(x) = x$ and $x \preceq x$ by
          assumption. 
        \item Transitivity: Assume that $(i, x) \preceq (j, y)$ and $(j, y) \preceq
          (k, z)$. Then,
          \[
            P_{ik}(x) = P_{jk} \circ P_{ij}(x) \preceq P_{jk}(y) \preceq z 
          \]
      \end{enumerate}
    \item Assume that $(I, \leq)$ is partially ordered as well as every
      $(P_i, \leq)$. Let $(i, x)$ and $(j, y)$ be elements of $\int P$ with 
      \[
        (i, x) \preceq (j, y)
        \quad\text{and}\quad
        (j, y) \preceq (i, x).
      \]
      Then, by definition, we have $i \preceq j$ and $j \preceq i$ so that $i =
      j$. Hence, 
      \[
        x = P^i_i(x) \leq y
        \quad\text{and}\quad
        y = P^i_i(y) \leq x
      \]
      implies that $x = y$ and thus $(i, x) = (j, y)$. 

    \item By definition. 
    \item Consider $(i, x) \in \int P$ and $j \in I$ with $i \preceq j$.
      Let $y \defeq P_{ij}(x)$. Then, obviously, $(i, x) \preceq (j, y)$
      by construction. Moreover, given $(k, z)$ with $(i, x) \preceq (k, z)$ and
      assuming $j \preceq k$, it follows that 
      \[ 
        P_{jk}(y) \cong P_{jk} \circ P_{ij}(x) \cong P_{ik}(x) 
        \preceq z 
      \]
      Therefore, $(j, y) \preceq (k, z)$. 
  \end{enumerate}
\end{proof}
\begin{remark}
  The construction presented here is a form of Grothendieck construction.  The
  full construction works for not only preordered sets but also categories and
  beyond. See, e.g., \cite{Jacobs1999}, for details.
\end{remark}

\subsection{Grothendieck construction for OPCMs} 
In this section, we extend the Grothendieck construction to OPCMs indexed by
a $\vee$-semilattice $(L, \preceq)$, where $L$ is partially ordered with a least
element denoted by $\bot$ and for every pair $(i, j)$ of elements there is a
least upper bound $i \vee j$. Given a (pseudo-)functor from $(L, \leq)$ to
$\OPCM$ we extend the local combination operations $\oplus_i$ for each $i \in L$ to
a global combination operation $\boxplus$ for $\int M$.

To simplify our discussion, we confine ourselves to functors instead of
pseudofunctors. Indeed, all of our discussion and examples in the remaining
section do not require this generality. 
\begin{theorem} \label{prop:grothendieck-2}
  Let $(L, \leq)$ be a bounded $\vee$-semilattice and
  $M\colon (L, \leq) \to \OPCM$ a functor. Then, the Grothendieck 
  completion $(\int M, \preceq)$ can be equipped with an OPCM given by
  \[
    (i, x) \boxplus (j, y) \defeq (k, M^i_k(x)  \oplus M^j_k(y) )
    \quad\text{and}\quad
    0 \defeq (\bot, 0_\bot)
  \]
  where $k = i \vee j$ and $(i, x) \boxplus (j, y)$ is defined if $M^i_k(x)
  \oplus M^j_k(y)$ is defined. 
\end{theorem}

\begin{proof}
  By \Cref{prop:grothendieck-1}, we know that $(\int M, \preceq)$ is a
  preordered set. We check each condition of OPCM as follows:
  \begin{enumerate}
    \item $(\bot, 0_\bot) \boxplus (i, x) = 
      (\bot \vee i, M^\bot_{\bot \vee i}(0_\bot) \oplus M^i_{\bot\vee i}(x))
      = (i, 0_{i} \oplus \id(x))
      = (i, x)$ since $M^\bot_{\bot \vee i}$ is a homomorphism, $\bot \vee i =
      i$, and $M^i_i = \id$. 

    \item $(i, x) \boxplus (j, y) = (i \vee j, M^i_k(x)  \oplus M^j_k(y) )
      = (j \vee i, M^j_k(y) \oplus M^i_k(x))
      = (j, y) \boxplus (i, x)$. 
    \item Similarly, the associativity follows routinely. 

    \item Assume that $(i_1, x_1) \preceq (i_2, x_2)$ and $(i_k, x_k) \boxplus
      (j, y)$ is defined for $k = 1, 2$. Let $i'_k \defeq i_k \vee j$ for $k =
      1, 2$. Then, by assumption, $i_1 \leq i_2$ so $i'_1 \leq i'_2$. 
      By definition, we have to prove 
      \begin{align*}
        (i_1, x_1) \boxplus (j, y)
        & = \left(i'_1, M^{i_1}_{i'_1}(x_1) \oplus M^j_{i'_1}(y)\right) \\ 
        & \preceq (i_2, x_2) \boxplus (j, y) = \left(i_2', M^{i_2}_{i'_2}(x_2) \oplus M^j_{i_2'}(y)\right)
      \end{align*}
      which is equivalent to prove
      \[
        M^{i_1'}_{i_2'}\left(M^{i_1}_{i_1'}(x_1) \oplus M^j_{i_1'}(y)\right)
        = M^{i_1}_{i_2'}(x_1) \oplus M^j_{i_2'}(y)
        \preceq M^{i_2}_{i_2'}(x_2) \oplus M^j_{i_2'}(y).
      \]
      To prove the above equation, it suffices to show that $M^{i_1}_{i_2'}(x_1)
      \preceq M^{i_2}_{i_2'}(x_2)$ since $\oplus$ is order-preserving. However,
      we know by assumption 
      \[
        M^{i_1}_{i_2}(x_1) \preceq x_2
      \]
      so by monotonicity of $M^{i_2}_{i'_2}$ it follows that 
      \[
        M^{i_1}_{i'_2}(x_1) = M^{i_2}_{i'_2} \circ M^{i_1}_{i_2}(x_1) \preceq 
        M^{i_2}_{i'_2}(x_2).
      \]
      Therefore, we have shown that $\boxplus$ is order-preserving. 
  \end{enumerate}
  
\end{proof}
The above construction is a slight modification of a form of Grothedieck
construction for monoidal categories, see~\cite{Shulman2008a} for details.

\subsection{Example: Natural join for relational dataset} 
Before we show our general result of ordered valuation algebras, we proceed with
our simplest example---the possibility representation. The linkage operation
$\boxplus$ derived from \Cref{prop:grothendieck-2} is the \emph{natural join} in
relation database theory~\cite{Codd1990}.

First of all, we assume that there is a set $\mathfrak{A}$ of known attribute
names and a set $\Phi_a$ of values for each attribute $a \in \mathfrak{A}$.  For
example, $\mathfrak{A}$ may consist of tags for UK postcode, personal
information, medical conditions, and so on.  By abuse of notation, we denote by
$\Phi_A$ for $A \subseteq \mathfrak{A}$ the cartesian product $\Phi_A \defeq
\prod_{a \in \mathfrak{A}} \Phi_a$.
Whenever $A \subseteq B$, we have projections $p_{B, A}$ from $\Phi_B$ to
$\Phi_A$ which sends $(x_b)_{b \in B}$ to $(x_a)_{a \in A}$. A functor
$P$ from the powerset $\Pow(\mathfrak{A}, \subseteq)$ to $\OPCM$ is
defined by
\[ (A \subseteq \mathfrak{A}) \mapsto (\Pow^+\Phi_A, \supseteq, \cap, \Phi_A)
  \quad\text{and}\quad (A \subseteq B) \mapsto \left(
  p_{B, A}^{-1}\colon \Pow^+\Phi_A \to \Pow^+\Phi_B\right).
\]

 In our interpretation, any set $S \in \Pow^+\Phi_A$ is a set
of possibilities where only one of them is true, so having more elements in $S$
means less specific information. If $A \subseteq B$, then $p_{B, A}^{-1}(S)$ is merely the
set $S$ padded with all combinations, i.e.\ $S \times \prod_{b \in B - A}
\Phi_b$.  So, $p_{B, A}^{-1}(S)$ contains no information about attributes $B -
A$.

Therefore the ordering on the Grothendieck completion $\int \Phi$
\[
  (A, S) \leq (B, T) \iff A \subseteq B \;\text{and}\; S \times \prod_{b \in B -
  A} \Phi_b \supseteq T
\]
simply means that $(A, S)$ is less informative than $(B, T)$ if $(B, T)$
contains more attributes and more specific on those already known in $A$.

By \Cref{prop:grothendieck-2}, the derived operation $\boxplus$ is given as
$(A, S) \boxplus (B, T) = (A \cup B, S \bowtie T)$ for $A, B \subseteq
\mathfrak{A}$, $S \in \Pow^+(\Phi_A)$, and $T \in \Pow^+(\Phi_B)$ where 
\[
  S \bowtie T = \set{ x \in \prod_{a \in A \cup B} \Phi_a }{ p_{A \cup B,
    A}(x) \in S \land p_{A\cup B, B}(x) \in T}
\]
which is by definition the natural join in relational database theory.

\subsection{Ordered valuation algebras}

It is observed in the community of approximate reasoning that with two algebraic
operations of combination and marginalisation a number of approximating
inference techniques can be formalised under reasonable assumptions. The
axiomatic approach is pursued by Shenoy and Shafer~\cite{Shenoy1990}, Shenoy and
Kohlas~\cite{Kohlas2000}, Haenni~\cite{Haenni2004}, etc. In this section, we
show that a variant of their axiomatic frameworks can be derived by our
Grothendieck construction for ordered commutative monoids, clarifying the
relationship between our approach and theirs. 
  
The following concept is derived from \cite{Haenni2004}:
\begin{definition}\label{def:valuation-algebra}
  A \emph{(stable) ordered valuation algebra} is a two-sorted algebra $(\Phi,
  \leq, D)$, consisting of a partially ordered set $(\Phi, \leq)$ of valuations
  and a bounded lattice $D$ of domains with operations \begin{enumerate} \item
    $\otimes\colon \Phi \times \Phi \to \Phi$ called \emph{combination}, \item
    $d\colon \Phi \to D$ such that $d(\varphi)$ is called the \emph{domain} of
    $\varphi$, \item $(-)^{\downarrow -}\colon \Phi \times D \rightharpoonup
    \Phi$ called \emph{focusing} where $\varphi^{\downarrow x}$ is defined for
    $x \leq d(\varphi)$, \item and $e\colon D \to \Phi$ such that $e_x$ is
  (called) an \emph{identity element} \end{enumerate} satisfying conditions
  below. In the following context, $\Phi_x = \set{ \varphi \in \Phi }{ d(\phi) =
  x }$.  \begin{enumerate} \item \label{def:valuation-algebra-2} $(\Phi,
    \otimes)$ is a commutative semigroup.

    \item \label{def:valuation-algebra-2.2} Comparable valuations are of the
      same domain: $\varphi \leq \psi$ implies $d(\varphi) = d(\psi)$.

    \item \label{def:valuation-algebra-3} Identity element: $d(e_x) = x$, $e_x
      \otimes e_y = e_{x \vee y}$, and $\varphi \otimes e_x = \varphi$ for
      $\varphi \in \Phi_x$.

    \item \label{def:valuation-algebra-3.5} Stability of identity under
      focusing: $e_y^{\downarrow x} = e_x$ for $x \leq y$. 

    \item \label{def:valuation-algebra-4} Labelling: $d(\varphi \otimes \psi) =
      d(\varphi) \vee d(\psi)$ and $\varphi^{{\downarrow} x} \in \Phi_x$. 

    \item \label{def:valuation-algebra-5} Transitivity of focusing
      $(\varphi^{\downarrow y})^{\downarrow x} = \varphi^{{\downarrow} x}$ for
      $x \leq y \leq d(\varphi)$.

    \item \label{def:valuation-algebra-6} Distributivity of focusing over
      combination: $(\varphi \otimes \psi)^{\downarrow d(\varphi)} = \varphi
      \otimes \psi^{{\downarrow} d(\varphi) \wedge d(\psi)}$.

    \item \label{def:valuation-algebra-7} Combination preserves ordering:
      $\varphi_1 \otimes \varphi_2 \leq \psi_1 \otimes \psi_2$ whenever
      $\varphi_i \leq \psi_i$.

    \item \label{def:valuation-algebra-8} Focusing preserves ordering:
  $\varphi^{\downarrow x} \leq \psi^{\downarrow x}$ for any $x \leq d(\varphi) =
  d(\psi)$ and $\varphi \leq \psi$.  \end{enumerate} \end{definition} The
  focusing operation $\downarrow$ formalises marginalisation in probability
  theory and projection in relational database theory. The intuitive meaning of
  every other operation is self-evident.  In addition to the focusing operation,
  an \emph{vacuous extension} operation, coined in~\cite{Kohlas2003}, $\uparrow y\colon \Phi_x \to \Phi_y$
  can be defined every $y \geq x$ via
  \[ 
    \varphi^{\uparrow y} \defeq \varphi \otimes e_y
  \]
  We will
  see that $\downarrow$ and $\uparrow$ forms a Galois connection under mild
  conditions. 

 \begin{remark} The original formulation in \cite{Haenni2004} imposes additional
   requirements. For example, $D$ is only a powerset instead of a lattice and
   $\Phi_x$ also requires a null element. More variants of (unordered) valuation
   algebras are discussed in~\cite{Pouly2011,Kohlas2003}.
 \end{remark}

\begin{proposition} \label{prop:valuation-algebras} Let $(\Phi, \leq, D;
  \otimes, \downarrow, e)$ be an ordered valuation algebra. Then, the following
  statements hold: \begin{enumerate} \item $(\Phi_x, \leq, \otimes, e_x)$ is an
    ordered commutative monoid.

    \item For any $x \leq y$, the vacuous extension operation $(-)^{\uparrow y}$
      is an order-preserving monoid homomorphism from $\Phi_x$ to $\Phi_y$.
      
    \item $(\Phi, \leq, D; \otimes, \downarrow)$ gives rise to a functor from
      $D$ to the category of ordered commutative monoids.

  \end{enumerate} \end{proposition} \begin{proof} \begin{enumerate} \item Assume
  that $\varphi, \psi \in \Phi_x$. Then, $ d(\varphi \otimes \psi) = d(\varphi)
  \vee d(\psi) = x \vee x = x$, so $\Phi_x$ is closed under combination
  $\otimes$. Therefore, it is easy to check that $(\Phi_x, \otimes, e_x)$ is an
  ordered commutative monoid. 

    \item Let $y \leq x \in D$. Then, we show that the mapping $\varphi \mapsto
      \varphi^{\uparrow y}$ is a monoid homomorphism: \begin{enumerate} \item
        $(e_x)^{\uparrow y} = e_x \otimes e_y = e_{x \vee y} = e_y$, and

        \item for any $\varphi$ and $\psi$ in $\Phi_x$, the following equations
          hold: \begin{align*} (\varphi \otimes \psi)^{\uparrow y} & = (\varphi
            \otimes \psi) \otimes e_y \\ & = \varphi \otimes \psi (\otimes e_y
            \otimes e_y) \\ & = (\varphi \otimes e_y) \otimes (\psi \otimes e_y)
            \\ & = \varphi^{\uparrow y} \otimes \varphi^{\uparrow y}
          \end{align*} by commutativity, associativity, and the identity element
          $e_y$. 

        \item Since $\otimes$ preserves the ordering, it then follows from the
          definition that \[ \varphi \leq \psi \implies \varphi \otimes e_y \leq
          \psi \otimes e_y \] that is, $\varphi^{\uparrow y} \leq \psi^{\uparrow
        y}$, since $e_y \leq e_y$.  \end{enumerate}

    \item The first two statements already show that there is a $D$-indexed
      family of ordered commutative monoids and for any $x \leq y$ there is an
      order-preserving monoid homomorphism $\Phi^x_{y} = (-)^{\uparrow y}$. It
      remains to show functoriality: $\Phi^x_x = \id$ and $\Phi^y_z \circ
      \Phi^x_y = \Phi^x_z$.  \begin{enumerate} \item $\Phi^x_x$ is evident:
        $\varphi \otimes e_x = \varphi$ for any $\varphi \in \Phi_x$, since
        $e_x$ is an identity.  \item Suppose that $x \leq y \leq z$. Then \[
          \Phi^y_z\circ\Phi^x_y(\varphi) = (\varphi \otimes e_y) \otimes e_z =
          \varphi \otimes (e_y \otimes e_z) = \varphi \otimes e_{y \vee z} =
          \varphi \otimes e_z = \varphi^{\uparrow z}.  \] \end{enumerate}
      \end{enumerate} \end{proof}

 As we intend to view ordered valuation algebras as Grothendieck completions of
 families of commutative monoids, an obvious discrepancy is that $\varphi$ and
 $\psi$ are comparable only if $d(\varphi) = d(\psi)$ in ordered valuation
 algebras while elements $(x, \varphi)$ and $(y, \psi)$ in $\int P$ are
 comparable even if domains $x$ and $y$ are different.  This can be
 readily mitigated by extending $\leq$ canonically: \[ \varphi \leq' \psi \iff
     d(\varphi) \leq d(\psi) \quad\text{and}\quad \varphi \otimes e_{d(\psi)}
   \leq \psi.  \]

 \begin{proposition} The ordered algebraic structure $(\Phi, \leq', D; \otimes,
   d, \downarrow, e)$ satisfies conditions\footnote{The order-preservation
     property of focusing accordingly becomes `if $\varphi \leq \psi$ and $x
     \leq d(\varphi)$ then $\varphi^{\downarrow x} \leq \psi^{\downarrow x}$'.}
     of ordered valuation algebra except that $\varphi \leq
     \psi$ implies $d(\varphi)= d(\psi)$.
 \end{proposition}
 \begin{proof}
   As the algebraic equations still hold, we only need to show conditions about
   the ordering, i.e.\ combination and marginalisation preserve partial order:
   \begin{enumerate} \item Assume that $\varphi_i \leq' \psi_i$ for $i = 1,
       2$. And let $y_i = d(\psi_i)$ and $y = y_1 \vee y_2 = d(\psi_1 \otimes
       \psi_2)$.  By assumption $\varphi_i \leq' \psi_i$ and the
       order-preservation property of $\otimes$, it follows that \[ (\varphi_1
         \otimes \varphi_2) \otimes e_{y} = (\varphi_1 \otimes \varphi_2)
         \otimes (e_{y_1} \otimes e_{y_2}) = (\varphi_1 \otimes e_{y_1}) \otimes
       (\varphi_2 \otimes e_{y_2}) \leq \psi_1 \otimes \psi_2 \] Therefore
       $\varphi_1 \otimes \varphi_2 \leq' \psi_1 \otimes \psi_2$.

     \item Assume that $\varphi \leq \psi$ and $x \leq d(\varphi)$ and let $y =
       d(\psi)$.  Observe that by transitivity of marginalisation and partial
       distributivity, \[ (\varphi \otimes e_y)^{\downarrow x} = \left((\varphi
         \otimes e_y)^{\downarrow d(\varphi)}\right)^{\downarrow x} =
       \left(\varphi \otimes e_y^{\downarrow d(\varphi)}\right)^{\downarrow x} =
     \left( \varphi \otimes e_{d(\varphi)} \right)^{\downarrow x} =
   \varphi^{\downarrow x} \] Therefore $\varphi \leq' \psi$, equivalently
   $\varphi \otimes e_y \leq \psi$, implies $\varphi^{\downarrow x} \leq'
   \psi^{\downarrow x}$ for any $x \leq d(\varphi)$, i.e. the focusing operation
   $\downarrow$ is order-preserving.  (Note that $\leq$ and $\leq'$ coincide for
   $\varphi, \psi$ with $d(\varphi) = d(\psi)$.)
   \end{enumerate}
 \end{proof}

By applying the Grothendieck construction (\Cref{prop:grothendieck-2}) to the
$D$-indexed family of ordered commutative monoids~$\Phi_x$
(\Cref{prop:valuation-algebras}), we have a partially ordered set $(\int \Phi,
\preceq)$.
The mapping  $(x, \varphi) \mapsto \varphi$ is evidently bijective since
$d(\varphi) = x$, and $(x, \varphi) \preceq (y, \psi) \iff \varphi \leq' \psi$ by
  definition. That is, the bijection $(x, \varphi) \mapsto \varphi$ is an order
  isomorphism between $(\int \Phi, \preceq)$ and $(\Phi, \leq')$.

It is clear that the domain operation $d\colon \Phi \to D$ is the projection
$p\colon \int \Phi \to D$ through the isomorphism, i.e.\ $p(x, \varphi) =
d(\varphi)$.  Similarly, $e_x \in \Phi_x$ is unique for each $x$, so it defines
$e\colon D \to \int \Phi$.

As for the combination operations $\otimes$ and $\boxtimes$, note that
$\boxtimes$ is given by \[ (x, \varphi) \boxtimes (y, \psi) = \left(z,
\varphi^{\uparrow z} \otimes \psi^{\uparrow z} \right) \] where $z = x \vee y$
and $\varphi^{\uparrow z} \otimes \psi^{\uparrow z} = \varphi \otimes \psi$ by
an easy calculation. Henceforth, $\otimes$ is the same as $\boxtimes$ via the
isomorphism.

It remains to derive the focusing operation from the Grothendieck construction.
To this point, we need a regularity condition: \begin{lemma}
  \label{lem:galois-connection} For any ordered valuation algebra $\Phi = (\Phi,
  \leq, D; \otimes, d, \downarrow, e)$, the following statements are true:
  \begin{enumerate} \item $\varphi^{\uparrow y} \leq \psi$ implies $\varphi \leq
    \psi^{\downarrow x}$.  \item If $e_x \leq \varphi$ for any $\varphi \in
    \Phi_x$ and $\Phi$ is \emph{regular},
      i.e.\ for any $\varphi$ and $x \leq d(\varphi)$ there is $\chi \in \Phi_x$
      such that $\varphi^{\downarrow x} \otimes \chi \otimes \varphi \leq
      \varphi$, then $\varphi \leq \psi^{\downarrow x}$ implies
      $\varphi^{\uparrow y} \leq \psi$.  \end{enumerate} \end{lemma}
      \begin{proof} \begin{enumerate} \item Assume that $\varphi^{\uparrow y}
        \leq \psi$ or equivalently $\varphi \otimes e_y \leq \psi$.  By
        stability, it follows that $(\varphi \otimes e_y)^{\downarrow x} =
        \varphi \otimes e_y^{\downarrow x} = \varphi \otimes e_x = \varphi$.
        Therefore, \[ \varphi \otimes e_y \leq \psi \implies (\varphi \otimes
          e_y)^{\downarrow x} = \varphi \leq \psi^{\downarrow x}.  \]

    \item Assume $\varphi \leq \psi^{\downarrow x}$ and by regularity there is
      $\chi \in \Phi_x$ such that $\psi^{\downarrow x} \otimes \chi \otimes \psi
      \leq \psi$. Then, \[ \varphi^{\uparrow y} = \varphi \otimes e_y \leq
        \psi^{\downarrow x} \otimes (e_y \otimes e_y) \leq \psi^{\downarrow x}
    \otimes \chi \otimes \psi \leq \psi.  \] \end{enumerate} \end{proof}
    \begin{remark} The condition(s) in \Cref{lem:galois-connection} are studied
      in \cite{Pouly2011}. Idempotent valuation algebras are called
      \emph{information algebra} by Kohlas \cite{Kohlas2003}.
    \end{remark}

Every adjoint is uniquely determined by the other adjoint, so in particular the
focusing operation $\downarrow$ is uniquely determined by the vacuous extension
$\uparrow$. 

To sum up, we have shown that the combination operation $\otimes$ of an ordered
valuation algebra can be derived by the Grothendieck construction:
\begin{theorem} \label{thm:valuation-algebra}
  Every regular ordered valuation algebra $(\Phi, \leq, D; \otimes, \downarrow, e)$ with
  $e_x \leq \varphi$ for any $\varphi \in \Phi_x$ is isomorphic to the Grothendieck
  completion $(\int \Phi, \preceq, \boxtimes, 0)$ of the functor given by
  \Cref{prop:valuation-algebras}.
\end{theorem}

\begin{remark}
  Both of \Cref{prop:grothendieck-2} and \Cref{thm:valuation-algebra} justify
  our claim that data linkage is made of data combination and changes of domain.
  The Grothendieck construction is in fact an equivalence of categories so that
  a pseudo-functor from a preorder to monoidal structures is essentially an
  opfibration equipped with a global monoidal structure. For interested readers,
  see~\cite[Theorem~12.7]{Shulman2008a}.
\end{remark}

\section{Concluding remarks}
Ubiquitous computing has led to ubiquitous data. Technologies exist that explore 
information content by combining data in a dataset and, in particular, linking
data from different datasets. Given the diversity of what passes for data---exact,
approximate, erroneous, fictitious---a very abstract conceptual framework is needed to
discover any general principles in today's
\emph{datafest}.

We have presented an abstract algebraic framework based on axiomatic notions
that model a data source, data representations and their combination `$\oplus$',
a measure of information content `$\preceq$', and linkage between data sources.
By stripping down intuitions we have found that \emph{ordered partial
commutative monoids} provide algebraic structures to be found at the heart of
many quite disparate data sharing situations. 

Our next steps are to map the scope of ordered partial commutative
monoids by exploring new and various 
\begin{enumerate}
  \item types of data, especially those in approximate reasoning such as belief
    functions and those discussed in uncertainty reasoning~\cite{Halpern2003}, etc; 

  \item types of operations on and between our algebras.
  
\end{enumerate}
Interestingly, there does not seem to be much of a theory of ordered partial
commutative monoids so that, too, is something to do. 

\bibliography{reference.bib}

\end{document}